\documentclass[letterpaper]{article} 
\usepackage{aaai2026}  
\usepackage{times}  
\usepackage{helvet}  
\usepackage{courier}  
\usepackage[hyphens]{url}  
\usepackage{graphicx} 
\usepackage{natbib}  
\usepackage{caption} 
\usepackage{algorithm}
\usepackage{algorithmic}
\usepackage{enumitem}
\usepackage{xcolor}
\usepackage{amsthm}
\usepackage{amsmath}
\usepackage{amssymb}
\usepackage{mathtools}

\theoremstyle{definition}
\newtheorem{definition}{Definition} 

\theoremstyle{plain}
\newtheorem{theorem}{Theorem}       
\newtheorem{lemma}[theorem]{Lemma}  
\newtheorem{corollary}[theorem]{Corollary}

\theoremstyle{remark}
\newtheorem*{remark}{Remark}

\newcommand{\cmt}[1]{}
\DeclareMathOperator*{\argmin}{arg\,min}
\usepackage{newfloat}
\usepackage{listings}
\DeclareCaptionStyle{ruled}{labelfont=normalfont,labelsep=colon,strut=off} 
\floatstyle{ruled}
\newfloat{listing}{tb}{lst}{}
\floatname{listing}{Listing}
\title{Pacing Equilibria in Second-Price Auctions with Few Goods}
\author{
    Yiyang Huang\textsuperscript{\rm 1},
    Yonglei Yan\textsuperscript{\rm 1},
    Zihe Wang\textsuperscript{\rm 2}, 
    Zhengyang Liu\textsuperscript{\rm 1}\thanks{Zhengyang Liu is the corresponding author.}\\
}
\affiliations{
    \textsuperscript{\rm 1}Beijing Institute of Technology\\
    \textsuperscript{\rm 2}Renmin University of China\\
    
    \{huangyiyang,yanyonglei,zhengyang\}@bit.edu.cn, 
    wang.zihe@ruc.edu.cn
}

\usepackage{bibentry}

\begin{document}

\maketitle

\begin{abstract}
In this paper, we investigate the computation of second-price pacing equilibria (SPPEs), a foundational model in online advertising auctions. We present a polynomial-time algorithm for computing exact SPPEs in instances with a constant number of goods. Our core technique maps buyers' pacing multipliers to the highest bids on each good, effectively partitioning the parameter space into a set of distinct geometric cells. By enumerating these cells, we fix the relative ordering of the bids and reduce the problem of equilibrium computation to a linear feasibility program. Finally, we demonstrate that this tractability extends to large-scale markets with an arbitrary number of goods, provided the goods can be aggregated into a constant number of valuation types.
\end{abstract}

\section{Introduction}
Online advertising markets are often modeled as a sequence of auctions, where advertisers repeatedly bid based on their valuations while facing fixed budgets. In practice, the standard second-price auction model is insufficient: advertisers cannot simply bid their true values in every auction, as doing so would prematurely exhaust their budgets. A common mechanism to manage this issue is {\em pacing}. Each buyer $i$ is assigned a pacing multiplier $\alpha_i \in [0,1]$, and their bid on good $j$ is scaled down from $v_{ij}$ to the paced bid $\alpha_i v_{ij}$. This multiplier allows the buyer to control their total spending while preserving the mechanics of the second-price rule within each individual auction.

This dynamic motivates the study of second-price pacing equilibria (SPPE), where buyers choose pacing multipliers to manage budgets under the second-price rule. Since being formalized by~\citet{cksm:22}, pacing equilibria have become a cornerstone for algorithmic studies in budget-constrained advertising~\citep{BG:19, WYDK:23, BBF:24, LPSZ:24}.

However, computing SPPEs in general markets is fundamentally intractable. Recent complexity results \citep{CCR:23, CL:25} have firmly established that, barring additional structural constraints, we should not expect efficient algorithms for computing or even approximating SPPEs. This computational barrier raises a natural question: {\em which restricted market structures still allow for efficient computation?}
Recently, \citet{yan:26} made progress in this direction by providing an algorithm for instances where the number of buyers is constant, while leaving the regime with a constant number of goods as an open problem. This latter setting presents a fundamentally different and more severe challenge: because the number of pacing multipliers scales linearly with the buyer population, existing methods that rely on searching a constant-dimensional multiplier space are no longer viable. In this paper, we overcome this barrier. We demonstrate that bounding the number of goods provides sufficient structural geometry to recover tractability, even when the strategy space grows dynamically with the number of buyers.

\begin{theorem}\label{thm:main}
Given any instance of a second-price pacing game with a constant number $c$ of goods, there exists an algorithm that computes a second-price pacing equilibrium in polynomial time.
\end{theorem}


\subsection{Related Works}
Budget management via pacing is a central focus in modern ad exchanges. The SPPE framework, formalized by~\citet{cksm:22}, provides robust descriptive power for these dynamics and is widely adopted for evaluating automated bidding strategies~\citep{BG:19, WYDK:23, BBF:24, LPSZ:24}. It is worth noting the sharp contrast between auction formats: while pacing and throttling games in first-price auctions admit efficient polynomial-time solutions~\citep{BCCIJE:07, CKK:21}, second-price auctions introduce non-linearities that complicate equilibrium computation.

For general second-price auctions, the intractability of SPPE is well-documented. \citet{CCR:23} proved that finding an SPPE is PPAD-complete even for inverse-polynomial approximations. The inapproximability bounds were further improved by~\citet{CL:25}, establishing PPAD-hardness for securing a $\gamma$-approximate SPPE for any constant $\gamma < 1/3$. These hardness results confirm that universal, efficient algorithms for arbitrary SPPE instances are highly impossible.

\paragraph{Comparison with \citet{yan:26}.}
Computing an SPPE fundamentally requires resolving the circular dependency between the allocation of goods and the buyers' pacing multipliers. This is challenging because the two are tightly coupled: allocations depend on the relative ordering of paced bids $\alpha_i v_{ij}$, while multipliers depend on the buyers' resulting total payments. Furthermore, the ``no unnecessary pacing'' condition (see Condition (d) in Definition~\ref{def:PE}) introduces a highly discrete component, requiring the algorithm to definitively determine whether each buyer is paced ($\alpha_i < 1$) or unpaced ($\alpha_i = 1$). 

Methodologically, our approach and that of \citet{yan:26} both build upon a \emph{cell decomposition} framework, a technique with a proven track record in solving complex market equilibria and optimal pricing problems \cite{devanur2008market, chen2018complexity, alaei2017computing}. The shared principle involves partitioning a continuous parameter space into a finite set of geometric cells. Within any given cell, discrete structural properties, such as the relative ordering of paced bids and the pacing status of individual buyers, are strictly fixed. This transformation reduces the highly non-linear equilibrium conditions to a sequence of tractable linear feasibility programs.

However, to address the complementary regime with a constant number of goods, we introduce a fundamentally different spatial representation. Our research diverges from \citet{yan:26} in the following three key aspects:

\begin{itemize}
    \item \textbf{Dimensionality of the Search Space:} In the constant-buyer regime, \citet{yan:26} directly decompose the $n$-dimensional multiplier space ($\alpha$-space), or equivalently, the $(n-1)$-dimensional ratio space $\left(\frac{\alpha_1}{\alpha_n}, \ldots, \frac{\alpha_{n-1}}{\alpha_n}\right)$. This space is partitioned by hyperplanes of the form $\alpha_i v_{ij}=\alpha_{i'}v_{i'j}$.For a constant number of buyers $n$, this generates only a polynomial number of cells. However, as $n$ increases, the hyperplane arrangement in this high-dimensional space produces an exponential number of cells, leading to a severe curse of dimensionality. 
    
    To bypass this bottleneck, we summarize the outcome of the pacing game using a vector $\lambda=(\lambda_1,\ldots,\lambda_c)$, where $\lambda_j$ records the highest paced bid on good $j\in[c]$. By decomposing this $c$-dimensional $\lambda$-space instead of the $\alpha$-space, we ensure the dimension remains strictly bounded, yielding a finite, polynomial number of geometric cells.
    
    \item \textbf{State Determination vs. Exhaustive Search:} To handle the ``no unnecessary pacing'' condition, \citet{yan:26} explicitly enumerates all possible discrete states of the buyers. While manageable for constant buyers (generating $2^c$ states), this approach requires an intractable exhaustive analysis over $2^n$ cases for general markets. Conversely, our algorithm leverages the geometric properties of the $\lambda$-space. Once the highest bid levels $\lambda$ are fixed, each buyer's pacing multiplier can be deterministically recovered as the largest multiplier consistent with these upper bounds. Inside any fixed cell $F$, the expression for each $\alpha_i$ resolves to a fixed linear function of $\lambda$. We dynamically derive the pacing state of every buyer directly from the geometric boundaries of the cell, unifying the previously discrete search into a single linear feasibility system.
    
    \item \textbf{Scalability and Market Realism:} The constant-buyer assumption in \citet{yan:26} is often at odds with real-world advertising markets, which typically feature thousands of competing advertisers. Their algorithm is essentially restricted to small-scale scenarios. In contrast, our framework naturally accommodates an arbitrary number of buyers. This reflects the reality of modern ad exchanges, where  impressions are numerous, they can often be aggregated into a constant number of distinct valuation types based on user segments or placement attributes.
\end{itemize}
\section{Preliminaries}

We recall the model of second-price pacing games and the equilibrium notion used throughout the paper.

\subsection{Second-Price Pacing Games}

A \emph{second-price pacing game} $G = (n, m, (v_{ij}), (B_i))$ involves $n$ buyers and $m$ indivisible goods. Each buyer $i$ has a valuation $v_{ij}$ for good $j$ and a strictly positive budget $B_i > 0$. In this game, each buyer $i$ chooses a \emph{pacing multiplier} $\alpha_i \in [0,1]$, effectively submitting a paced bid of $\alpha_i v_{ij}$ for each good $j$.

Goods are allocated via second-price auctions based on these paced bids. Let $h_j(\alpha) = \max_i \alpha_i v_{ij}$ denote the highest bid for good $j$, and $p_j(\alpha)$ denote the second-highest bid. The good is awarded to the highest bidder at price $p_j(\alpha)$. If multiple buyers tie for the highest bid, the good is allocated fractionally among them. In such tie-breaking cases, the unit price is simply $p_j(\alpha) = h_j(\alpha)$. We denote this fractional allocation by $x_{ij} \in [0,1]$, which can be interpreted as the probability that buyer $i$ wins the indivisible good $j$.

We now formally define a pacing equilibrium in second-price auctions.

\begin{definition}[Pacing Equilibrium]
\label{def:PE}
We say $(\alpha, x)$ with $\alpha = (\alpha_i) \in [0,1]^n$, $x = (x_{ij}) \in [0,1]^{nm}$, and $\sum_{i\in [n]} x_{ij} \le 1$ for all $j \in [m]$, is a {\em pacing equilibrium} of a second-price pacing game $G = (n, m, (v_{ij}), (B_i))$ if:
\begin{enumerate}[label=\normalfont(\alph*), leftmargin=*, itemsep=0pt]
    \item $x_{ij} >0$ implies $\alpha_i v_{ij} = h_j(\alpha)$;
    \item $h_j(\alpha) > 0$ implies $\sum_{i\in [n]} x_{ij} = 1$;
    \item $\sum_{j\in [m]} x_{ij}p_j(\alpha) \le B_i$;
    \item $\sum_{j\in [m]} x_{ij}p_j(\alpha) < B_i$ implies $\alpha_i = 1$.
\end{enumerate}
\end{definition}
We refer to Condition (d) as the ``no unnecessary pacing'' condition. Economically, this condition enforces that no buyer is subjected to excessive or unwarranted pacing. 

The existence of pacing equilibria was proved by \citet{cksm:22}.

\begin{theorem}[\citet{cksm:22}]
\label{thm:pacing-equilibrium-exists}
Any second-price pacing game admits a pacing equilibrium.
\end{theorem}

Previously speaking, our algorithm relies on a cell-decomposition argument in a constant-dimensional parameter space. We will use following theorem on the number of faces generated by a hyperplane arrangement.

\begin{theorem}[\citet{zaslavsky:75,orlik:92}]\label{thm:cut}
Let $\mathcal{A}$ be a collection of $k$ affine hyperplanes in general position in $\mathbb{R}^d$. Then the number of $j$-dimensional faces ($0 \le j \le d$) induced by the arrangement is given by
\[
f_j = \binom{k}{d-j} \sum_{i=0}^{j} \binom{k-d+j}{i}.
\]

In particular, the number of $d$-dimensional cells (i.e., the connected components of $\mathbb{R}^d \setminus \bigcup \mathcal{A}$) is
\[
f_d = \sum_{i=0}^{d} \binom{k}{i}.
\]
\end{theorem}

Before proceeding to the main proof, we introduce the necessary setup.
\subsection{Model Setup and $\lambda$-Space Reformulation.}

We also allow a dummy witness, indexed by \(0\), to represent the case where the second-price payment of a good is zero.

\begin{remark}
The dummy witness is only a notational device. It is not a real buyer, receives no allocation, and has no budget constraint. For every good \(j\in[m]\), we define \(v_{0j}=0\) and \(\alpha_0=1\), so its paced bid is always \(\alpha_0 v_{0j}=0\). Thus, choosing \(r_j=0\) simply means that the witness for the second price of good \(j\) is the dummy buyer, i.e., the second price is zero.
\end{remark}

For each good \(j\in [m]\), we define
\begin{equation}
\label{eq:lambda-definition}
\lambda_j :=\max_{i\in [n]}\alpha_i v_{ij},
\end{equation}
which denotes the highest paced bid on good \(j\). Note that \(\lambda_j\) serves as an upper bound for the bids and is not necessarily the true price of the good (which is determined by the second-highest paced bid). The vector \(\lambda=(\lambda_1,\ldots,\lambda_c)\) is central to our algorithm because it allows us to completely recover the pacing multipliers \(\alpha = (\alpha_1, \ldots, \alpha_n)\).

Before deriving the exact connection between \(\lambda\) and \(\alpha\), we establish two simplifying lemmas.

\begin{lemma}\label{lem:positive-alpha}
In any SPPE, every buyer has a strictly positive pacing multiplier. 
\end{lemma}

\begin{proof}
Suppose, for contradiction, that \(\alpha_i=0\) for some buyer \(i\in [n]\), so buyer \(i\)'s paced bids are all zero. Consequently, their total payment must be zero regardless of their allocation. Because their budget \(B_i>0\) is strictly positive, they do not exhaust their budget. The ``no unnecessary pacing'' condition then dictates that \(\alpha_i=1\), which contradicts \(\alpha_i=0\). Thus, \(\alpha_i>0\) for all \(i\in [n]\).
\end{proof}

\begin{lemma}\label{lem:remove-zero-lambda-goods}
In any SPPE, if there exists a good \(j'\in [m]\) with \(\lambda_{j'}=0\), removing \(j'\) preserves the equilibrium.
\end{lemma}

\begin{proof}
If \(\lambda_{j'}=0\), then \(\alpha_i v_{ij'}=0\) for all \(i\in [n]\). This implies that all buyers are tied for the highest paced bid of zero, and the second-price payment for \(j'\) is exactly zero. Because allocating this good incurs no payment, its allocation does not affect any buyer's budget constraints or pacing status. Thus, removing \(j'\), or conversely allocating it arbitrarily at a price of zero in a reduced instance, leaves all equilibrium conditions unchanged.
\end{proof}

By Lemma~\ref{lem:remove-zero-lambda-goods}, we can assume without loss of generality that \(\lambda_j > 0\) for all remaining goods. Specifically, since \(\alpha_i > 0\) (Lemma~\ref{lem:positive-alpha}), we simply remove any good \(j\) where \(v_{ij} = 0\) for all \(i \in [n]\). 

With these zero-value goods removed, we can now express each buyer's multiplier \(\alpha_i\) strictly in terms of \(\lambda\). By the definition of \(\lambda_j\), we have \(\alpha_i v_{ij}\leq \lambda_j\) for all goods \(j\). Combined with the baseline constraint \(\alpha_i\leq 1\), each buyer's pacing multiplier is bounded by:
\[
\alpha_i
\le
\min\left(
\{1\}\cup
\left\{
\frac{\lambda_j}{v_{ij}}: j\in [m],\ v_{ij}>0
\right\}
\right).
\]

We claim this upper bound must be tight. Suppose for contradiction that it is strict for some buyer \(i\). This implies \(\alpha_i < 1\), and crucially, \(\alpha_i v_{ij} < \lambda_j\) for every good \(j\) where \(v_{ij} > 0\). This means buyer \(i\) strictly fails to attain the highest paced bid on any good they value. Consequently, they receive no goods and pay zero. Since \(B_i>0\), their budget is not exhausted. The ``no unnecessary pacing'' condition then forces \(\alpha_i=1\), contradicting \(\alpha_i < 1\). 

Therefore, the upper bound is tight, and we obtain the deterministic formula for each $i\in[n]$:
\begin{equation}
\label{eq:alpha-lambda}
\alpha_i(\lambda)
=
\min\left(
\{1\}\cup
\left\{
\frac{\lambda_j}{v_{ij}}: j\in [m],\ v_{ij}>0
\right\}
\right).
\end{equation}

Eq.~\eqref{eq:alpha-lambda} allows us to recover the multiplier vector \(\alpha\) directly from the highest-bid vector \(\lambda\). Furthermore, once \(\lambda\) is fixed, the buyers' paced bids \(\alpha_i v_{ij}\) are also fixed, allowing us to immediately identify the highest bidders for each good.
\paragraph{Partitioning the \(\lambda\)-space.}
To partition the \(\lambda\)-space where \(\lambda_j > 0\) for all \(j\) by Lemma~\ref{lem:remove-zero-lambda-goods}, we observe that the expression for \(\alpha_i(\lambda)\) is piecewise linear. Its structural form changes only when two terms within the minimum function become equal. Thus, for each buyer \(i\), we define bounding hyperplanes by comparing the constant term \(1\) with each ratio \(\lambda_j/v_{ij}\), and by comparing \(\lambda_j/v_{ij}\) with \(\lambda_k/v_{ik}\).

\begin{remark}
Notably, our formulation explicitly restricts the comparison set to goods where \(v_{ij} > 0\). Because a buyer will never attain the highest paced bid on a good for which they hold zero value, these terms are safely excluded. Consequently, all denominators \(v_{ij}\) appearing in the hyperplane equations are strictly positive, ensuring the arrangement is mathematically well-defined.
\end{remark}

The comparison \( 1  = \lambda_j/v_{ij} \)
gives the hyperplane
\begin{equation}
\label{eq:coordinate-hyperplane}
\lambda_j=v_{ij},
\qquad 
\forall i\in [n],\ \forall j\in [m].
\end{equation}

Similarly, the comparison \(\lambda_j/v_{ij}=\lambda_k/v_{ik}\) gives the hyperplane
\begin{equation}
\label{eq:ratio-hyperplane}
\frac{\lambda_j}{\lambda_k}=\frac{v_{ij}}{v_{ik}},
\qquad
\forall i\in [n],\ \forall j,k\in [m],\text{ s.t. }\ j<k.
\end{equation}

These hyperplanes induce a decomposition of the \(\lambda\)-space into relatively open faces of various dimensions. We refer to all such relatively open faces as \emph{cells}, and let \(\mathcal{F}\) denote the collection of these cells.

\begin{remark}
We include lower-dimensional faces in \(\mathcal{F}\) because ties among paced bids may occur exactly on the boundaries of the hyperplane arrangement. Excluding these faces would miss equilibria involving tied highest bids.
\end{remark}

Within any fixed cell \(F\in\mathcal{F}\), the relative ordering of all terms in \(\alpha\)'s expression is strictly fixed for every buyer \(i\). Consequently, we can deterministically identify which term achieves the minimum, thereby determining the exact expression for \(\alpha_i\). Equivalently, this structural property allows us to express the multiplier vector \(\alpha\) as a fixed linear function of \(\lambda\).

Consequently, we now define the following set of minimizers:
\[
M_i(F)=
\argmin_{k\in \{0\}\cup\{j:v_{ij}>0\}} t_{ik}(\lambda),
\]
where
\[
t_{i0}(\lambda)=1,
\quad
t_{ij}(\lambda)=\frac{\lambda_j}{v_{ij}}
\qquad \text{for } j\in[m]\text{ with }v_{ij}>0.
\]

The set \(M_i(F)\) denotes the set of goods indices whose corresponding terms attain the minimum within the cell \(F\). In other words, for every good in \(M_i(F)\), the corresponding term can represent $\alpha_i$. If \(0\in M_i(F)\), then \(\alpha_i(\lambda)=1\); and \(\alpha_i(\lambda)=\lambda_{j} / {v_{ij}}\) for any \(j\in M_i(F)\), otherwise.

\begin{remark}
The index \(0\) in \(M_i(F)\) corresponds to the constant term \(1\) in the minimum defining \(\alpha_i(\lambda)\). Thus \(0\in M_i(F)\) means that the constraint \(\alpha_i\le 1\) is binding, and hence buyer \(i\) is unpaced.
\end{remark}

As above, within each fixed cell \(F\), \(M_i(F)\) remains unchanged, so \(\alpha\) admits a fixed linear representation in terms of \(\lambda\).

Similarly, for each fixed cell \(F\in\mathcal{F}\) and each good \(j\in [m]\), we define:
\[
T_j(F)
:=
\{\, i\in [n] : j\in M_i(F) \,\},
\]
which is exactly the set of buyers whose bid attains the highest paced bid level on good \(j\). Indeed, for any buyer \(i\) with \(v_{ij}>0\), we have
\[
j\in M_i(F)
\ \Longleftrightarrow\
\alpha_i(\lambda)=\frac{\lambda_j}{v_{ij}}
\ \Longleftrightarrow\
\alpha_i(\lambda)v_{ij}=\lambda_j.
\]
So \(T_j(F)\) is precisely the set of buyers whose paced bid on good \(j\) equals \(\lambda_j\). Within the cell \(F\), this set is also fixed.

\section{Proof of Theorem~\ref{thm:main}}
With the necessary foundations established, we now prove our main result. We first present our algorithm, followed by an analysis of its correctness and time complexity.

\subsection{Our Algorithm}
We now present our algorithmic framework for instances with a constant number \(c\) of goods, establishing Theorem~\ref{thm:main}. Our approach consists of two integrated procedures: Algorithm~\ref{alg:second-price-witness} acts as a crucial subroutine to systematically enumerate the valid second-price witness sets for each good, while our primary procedure, Algorithm~\ref{alg:constant-goods}, utilizes these sets to compute the exact pacing equilibrium.

\begin{algorithm}[H]
\caption{Finding a pacing equilibrium in second-price auctions with constant goods.}
\label{alg:constant-goods}
\begin{algorithmic}[1]
\REQUIRE Set of buyers \(N=\{1,\ldots,n\}\), goods \(M=\{1,\ldots,c\}\), budget \(B_i\) for each \(i\in [n]\), and valuation matrix \((v_{ij})\).
\ENSURE Pacing multipliers \((\alpha_i)\) and allocations \((x_{ij})\).

\STATE Introduce a dummy buyer.
\STATE Remove all \(j\) which \(v_{ij}=0\) for \( \forall i\in[n]\), i.e., \(\lambda_j = 0\).

\STATE Construct the following hyperplanes in the \(\lambda\)-space:
\[
\lambda_j=v_{ij},
\qquad
\forall i\in [n],\ \forall j\in [m],
\]
and
\[
\frac{\lambda_j}{\lambda_k}=\frac{v_{ij}}{v_{ik}},
\qquad
\forall i\in [n],\ \forall j,k\in [m],\ j<k.
\]

\STATE // Enumerate all cells of this hyperplane arrangement.
\FOR{each cell \(F \in \mathcal{F} \)}
    \FOR{every buyer \(i\in [n]\)}
        \STATE Construct \(M_i(F)\) and determine \(\alpha_i(\lambda)\).
    \ENDFOR

    \FOR{each good \(j\in [m]\)}
        \STATE Construct \(T_j(F)\).
        \STATE Use Algorithm~\ref{alg:second-price-witness} to determine the second-price witness set \(R_j(F)\).
        \IF{Algorithm~\ref{alg:second-price-witness} returns infeasible}
            \STATE Skip this \(F\) and continue to the next cell.
        \ENDIF
    \ENDFOR
    
    \FOR{each witness tuple \(r\)}
    \STATE Construct the corresponding linear system to find \((\lambda_j)\) and \((x_{ij})\).
    \IF{there is a feasible solution}
        \STATE Recover \((\alpha_i)\) and \((x_{ij})\), and return them.
    \ENDIF
\ENDFOR

\ENDFOR
\end{algorithmic}
\end{algorithm}

The concrete linear system used by the algorithm~\ref{alg:constant-goods} (Line 17) will be specified in our main proof of Theorem~\ref{thm:main}. 

\subsection{Correctness of Our Algorithm}
Recall our main theorem:

\noindent\textbf{Theorem~\ref{thm:main}.}
Given any instance of a second-price pacing game with a constant number \(c\) of goods, there exists an algorithm that computes a second-price pacing equilibrium in polynomial time.

\begin{proof}[Proof of Theorem~\ref{thm:main}]
By Theorem~\ref{thm:pacing-equilibrium-exists}, every instance of the second-price pacing game exists at least one SPPE. We now declare that Algorithm~\ref{alg:constant-goods} is exactly the algorithm mentioned in this theorem. So it remains to show that Algorithm~\ref{alg:constant-goods} finds such an SPPE in polynomial time.

It is sufficient to prove that, for each fixed cell \(F\in\mathcal{F}\), the algorithm can decide whether \(F\) contains an equilibrium in polynomial time. The polynomial bound on the number of cells is established in the time-complexity analysis below. If so, we recover one as an SPPE.

\begin{algorithm}[H]
\caption{Finding second-price witness set for a good.}
\label{alg:second-price-witness}
\begin{algorithmic}[1]
\REQUIRE A good \(j\), a cell \(F\), the set \(T_j(F)\), and the expressions
\(\alpha_i(\lambda)\) determined by \(F\).
\ENSURE A set \(R_j(F)\) of candidate second-price witnesses, or infeasible.

\IF{\(T_j(F)=\emptyset\)}
    \RETURN infeasible

\ELSIF{\(|T_j(F)|\geq 2\)}
    \STATE Set \( R_j(F)=T_j(F)\).
    \RETURN \(R_j(F)\)

\ELSE
    \STATE Let \(T_j(F)=\{w_j\}\).
    \STATE Initialize \(R_j(F)=\emptyset\).

    \FOR{each \(r_j\in([n]\setminus\{w_j\})\cup\{0\}\)}
        \STATE Consider the following linear inequality system:
        \[
        \alpha_{r_j}(\lambda)v_{r_jj}
        \geq
        \alpha_i(\lambda)v_{ij},
        \qquad
        \forall i\in[n]\setminus\{w_j\}.
        \]
        \IF{the above system is feasible together with the current cell constraints}
            \STATE Add \(r_j\) to \(R_j(F)\).
        \ENDIF
    \ENDFOR

    \IF{\(R_j(F)=\emptyset\)}
        \RETURN infeasible
    \ELSE
        \RETURN \(R_j(F)\)
    \ENDIF
\ENDIF

\end{algorithmic}
\end{algorithm}

Given a fixed cell \( F\), the set \(M_i(F)\) determines which term attains the minimum in the expression for \(\alpha_i(\lambda)\). In particular, if both \(1\) and \(\lambda_j/v_{ij}\) attain the minimum, then both \(0\) and \(j\) belong to \(M_i(F)\). In this way, such a buyer is treated as a buyer with \(\alpha_i(\lambda)=1\).
 
Following the case distinction in~\citet{yan:26} and by Lemma~\ref{lem:positive-alpha}, we can classify the buyers into two types, \(\alpha_i = 1\) and \( 0 < \alpha_i < 1\). Notably, without \(\lambda\), we cannot enumerate all possible classifications described classifying above in polynomial time. Indeed, there would be \(2^n\) possible patterns, which is not polynomial in \(n\), making their method impossible for constant goods instances. 

But after our transformation from \(\alpha\) to \(\lambda\), we can determine each buyer \(i\) in a fixed cell \(F\) is whether \(\alpha_i(\lambda)=1\) or \(0 < \alpha_i(\lambda)<1\) by checking the set \(M_i(F)\). Specifically, we have
\begin{equation}
\label{eq:alpha-fixed-cell}
\alpha_i(\lambda)=
\begin{cases}
1, & 0\in M_i(F),\\[4pt]
\dfrac{\lambda_{j}}{v_{ij}}, & 0\notin M_i(F),\ j\in M_i(F).
\end{cases}
\end{equation}

The ``no unnecessary pacing'' condition in the definition  imposes different requirements on buyers with \(\alpha_i=1\) and those with \(0 < \alpha_i<1\). In fact, if a buyer does not exhaust the budget, its pacing multiplier should be forced at 1, meaning that a buyer who does not exhaust its budget must remain unpaced.
 
In order to formulate the equilibrium conditions as a linear program, we introduce two sets to distinguish these two types of buyers. For each fixed cell \(F\), we define:
\[
\begin{aligned}
E(F) &:= \{i\in [n]:0\in M_i(F)\}, \\
L(F) &:= \{i\in [n]:0\notin M_i(F)\}.
\end{aligned}
\]

Equivalently,
\begin{align*}
E(F) &= \{i\in [n]:\alpha_i(\lambda)=1\}, \\
L(F) &= \{i\in [n]:\alpha_i(\lambda)<1\}. 
\end{align*}

By definition, buyers in \(E(F)\) are unpaced, and buyers in \(L(F)\) are paced. These two sets are determined directly by the cell \(F\).

We now construct the linear system for a fixed cell \(F\) type by type.

\paragraph{Cell constraints.}
Since \(F\) is a relatively open cell induced by our hyperplane arrangement, the sign of each defining comparison is fixed throughout \(F\). For every \(i\in[n]\) and every \(j\in[m]\) with \(v_{ij}>0\), the sign of \( (\lambda_j-v_{ij}) \) is fixed in \(F\). According to the cell \(F\), we impose one of
\[
\lambda_j < v_{ij},\qquad
\lambda_j = v_{ij},\qquad
\text{or}\qquad
\lambda_j > v_{ij}.
\]

Similarly, for every \(i\in[n]\) and every pair \(j,k\in[m]\) with \(j<k\) and \(v_{ij},v_{ik}>0\), the sign of \( (v_{ik}\lambda_j-v_{ij}\lambda_k) \) is fixed in \(F\). According to the cell \(F\), we impose one of
\[
v_{ik}\lambda_j < v_{ij}\lambda_k,\quad
v_{ik}\lambda_j = v_{ij}\lambda_k,\quad
\text{or}\quad
v_{ik}\lambda_j > v_{ij}\lambda_k.
\]

By Lemma~\ref{lem:remove-zero-lambda-goods}, after removing all goods with zero highest paced bid, we impose
\[
\lambda_j>0,
\qquad \forall j\in[m].
\]

For convenience, let \( A_F^{<},\ A_F^{=},\ A_F^{>} \)
be the partition of pairs \((i,j)\). This can be determined according to whether \(\lambda_j-v_{ij}\) is negative, zero, or positive throughout \(F\).

Similarly, let \( B_F^{<},\ B_F^{=},\ B_F^{>} \) be the partition of triples \((i,j,k)\), according to whether \(v_{ik}\lambda_j-v_{ij}\lambda_k\) is negative, zero, or positive throughout \(F\).

In general, the corresponding cell constraints are
\begin{equation}
\label{eq:cell-constraints}
\begin{dcases}
\lambda_j < v_{ij},
& \forall (i,j)\in A_F^{<},\\
\lambda_j = v_{ij},
& \forall (i,j)\in A_F^{=},\\
\lambda_j > v_{ij},
& \forall (i,j)\in A_F^{>},\\[1mm]
v_{ik}\lambda_j < v_{ij}\lambda_k,
& \forall (i,j,k)\in B_F^{<},\\
v_{ik}\lambda_j = v_{ij}\lambda_k,
& \forall (i,j,k)\in B_F^{=},\\
v_{ik}\lambda_j > v_{ij}\lambda_k,
& \forall (i,j,k)\in B_F^{>},\\[1mm]
\lambda_j > 0,
& \forall j\in[m].
\end{dcases}
\end{equation}

\paragraph{Witness and budget constraints.}
In Algorithm~\ref{alg:second-price-witness}, for every good \(j\in[m]\) we obtain a set \(R_j(F)\) of second-price witnesses.

We then enumerate one witness tuple
\[
r=(r_j)_{j\in[m]}\in \prod_{j\in[m]} R_j(F).
\]

For this tuple, the second-price payment of every good \(j\) is defined by the paced bid of its selected witness \(r_j\):
\[
p_j^{F,r}(\lambda)
=
\begin{cases}
\alpha_{r_j}(\lambda)v_{r_jj}, & r_j\in[n],\\[1mm]
0, & r_j=0.
\end{cases}
\]

When the cell \(F\) and the witness tuple \(r\) are clear from context, we simply use \( p_j(\lambda)\).

Each \(r_j\) is associated only with good \(j\): it is the selected witness whose paced bid determines the second-price payment of good \(j\). Thus, once the tuple \(r\) is fixed, each payment \(p_j(\lambda)\) is determined separately by the corresponding component \(r_j\).

We next impose the witness constraints. If \(T_j(F)=\{w_j\}\), \(w_j\) is the only highest paced bidder on good \(j\). In this case, our selected witness \(r_j\in R_j(F)\subseteq [n]\setminus\{w_j\}\) must have paced bid at least as large as every other losing bidder. We impose
\[
\alpha_{r_j}(\lambda)v_{r_jj}
\geq
\alpha_i(\lambda)v_{ij},
\qquad
\forall i\in[n]\setminus\{w_j\}.
\]

If \(|T_j(F)|\geq 2\), our selected witness \(r_j\in R_j(F)=T_j(F)\) already attains the highest paced bid on good \(j\), so no additional witness constraint is needed.

After that, we introduce payment variables
\[
y_{ij}\geq 0,
\qquad
\forall i\in[n],\ j\in[m],
\]
where
$y_{ij} = x_{ij} p_j$, for any $i\in[n],\ j\in[m]$,
which denotes the amount paid by buyer \(i\) for good \(j\).

A buyer can pay for good \(j\) only if the buyer is among the highest paced bidders for that good. Within the cell \(F\), this set is \(T_j(F)\). We impose
\[
y_{ij}=0,
\qquad
\forall j\in[m],\ \forall i\notin T_j(F).
\]

The total payment collected from good \(j\) must equal its second-price payment. For every \(j\in[m]\), we impose
\[
\sum_{i\in T_j(F)} y_{ij}
=
p_j(\lambda).
\]

The budget feasibility and ``no unnecessary pacing'' constraints force every unpaced buyer:
\[
\sum_{j\in[m]} y_{ij}\le B_i,
\qquad
\forall i\in E(F),
\]
and for every paced buyer,
\[
\sum_{j\in[m]} y_{ij}=B_i,
\qquad
\forall i\in L(F).
\]

In general, the corresponding witness and budget constraints are
\begin{equation}
\label{eq:witness-budget-constraints}
\begin{dcases}
\alpha_{r_j}(\lambda)v_{r_jj}
\geq
\alpha_i(\lambda)v_{ij},
& \begin{aligned}[t]
\forall j\in[m]\text{ with }T_j(F)=\{w_j\},\\
\forall i\in[n]\setminus\{w_j\},
\end{aligned}\\[1mm]
y_{ij}\geq 0,
& \forall i\in[n],\ \forall j\in[m],\\[1mm]
y_{ij}=0,
& \forall j\in[m],\ \forall i\notin T_j(F),\\[1mm]
\displaystyle
\sum_{i\in T_j(F)} y_{ij}
=
p_j(\lambda),
& \forall j\in[m],\\[1mm]
\displaystyle
\sum_{j\in[m]} y_{ij}\leq B_i,
& \forall i\in E(F),\\[1mm]
\displaystyle
\sum_{j\in[m]} y_{ij}=B_i,
& \forall i\in L(F).
\end{dcases}
\end{equation}

\paragraph{Other consistency constraints.}
Now we also impose several consistency constraints connecting \(\alpha\) and \(\lambda\). 

First, buyers in \(E(F)\) are unpaced, while buyers in \(L(F)\) are paced. We impose
\[
\begin{dcases}
\alpha_i(\lambda)=1,
& \forall i\in E(F),\\[1mm]
0<\alpha_i(\lambda)<1,
& \forall i\in L(F).
\end{dcases}
\]

Second, the vector \(\lambda\) must dominate every buyer's paced bid on every good. We impose
\[
\lambda_j
\geq
\alpha_i(\lambda)v_{ij},
\qquad
\forall i\in[n],\ \forall j\in[m].
\]

Last, for every buyer who belongs to the winning set \(T_j(F)\), the paced bid must exactly attain the highest paced bid level. We impose
\[
\lambda_j
=
\alpha_i(\lambda)v_{ij},
\qquad
\forall j\in[m],\ \forall i\in T_j(F).
\]

In general, the corresponding consistency constraints are
\begin{equation}
\label{eq:consistency-constraints}
\begin{dcases}
\alpha_i(\lambda)=1,
& \forall i\in E(F),\\[1mm]
0<\alpha_i(\lambda)<1,
& \forall i\in L(F),\\[1mm]
\lambda_j\geq \alpha_i(\lambda)v_{ij},
& \forall i\in[n],\ \forall j\in[m],\\[1mm]
\lambda_j= \alpha_i(\lambda)v_{ij},
& \forall j\in[m],\ \forall i\in T_j(F).
\end{dcases}
\end{equation}

\paragraph{Linear system.}
 Now combine Eqs.~\eqref{eq:cell-constraints}, \eqref{eq:witness-budget-constraints} and \eqref{eq:consistency-constraints}. For a fixed cell \(F\) and a fixed witness tuple \(r=(r_j)_{j\in[m]}\), the algorithm considers the following linear system:

\begin{equation}
\label{eq:linear-system}
\begin{dcases}
\lambda_j < v_{ij},
& \forall (i,j)\in A_F^{<},\\
\lambda_j = v_{ij},
& \forall (i,j)\in A_F^{=},\\
\lambda_j > v_{ij},
& \forall (i,j)\in A_F^{>},\\[1mm]
v_{ik}\lambda_j < v_{ij}\lambda_k,
& \forall (i,j,k)\in B_F^{<},\\
v_{ik}\lambda_j = v_{ij}\lambda_k,
& \forall (i,j,k)\in B_F^{=},\\
v_{ik}\lambda_j > v_{ij}\lambda_k,
& \forall (i,j,k)\in B_F^{>},\\[1mm]
\lambda_j > 0,
& \forall j\in[m],\\[1mm]
\alpha_i(\lambda)=1,
& \forall i\in E(F),\\
0<\alpha_i(\lambda)<1,
& \forall i\in L(F),\\[1mm]
\lambda_j\geq \alpha_i(\lambda)v_{ij},
& \forall i\in[n],\ \forall j\in[m],\\
\lambda_j= \alpha_i(\lambda)v_{ij},
& \forall j\in[m],\ \forall i\in T_j(F),\\[1mm]
\alpha_{r_j}(\lambda)v_{r_jj}
\geq
\alpha_i(\lambda)v_{ij},
& \begin{aligned}[t]
\forall j\in[m]\text{ with }T_j(F)=\{w_j\},\\
\forall i\in[n]\setminus\{w_j\},
\end{aligned}\\[1mm]
y_{ij}\geq 0,
& \forall i\in[n],\ \forall j\in[m],\\
y_{ij}=0,
& \forall j\in[m],\ \forall i\notin T_j(F),\\[1mm]
\displaystyle
\sum_{i\in T_j(F)} y_{ij}
=
p_j(\lambda),
& \forall j\in[m],\\[1mm]
\displaystyle
\sum_{j\in[m]} y_{ij}\leq B_i,
& \forall i\in E(F),\\
\displaystyle
\sum_{j\in[m]} y_{ij}=B_i,
& \forall i\in L(F).
\end{dcases}
\end{equation}

Strict inequalities are handled using a standard slack-variable transformation.

\begin{remark}
Strict inequalities in the cell constraints can be handled by introducing a slack variable \(\delta>0\) and replacing each strict inequality \(a^\top z<b\) with \(a^\top z+\delta\le b\). Since the number of constraints is polynomial, this preserves polynomial-time feasibility checking.
\end{remark}

\paragraph{Recovering the equilibrium.}
Suppose that the linear system~\eqref{eq:linear-system} is feasible. From a feasible solution \((\lambda,y)\), the algorithm outputs pacing multipliers
\[
\alpha_i=\alpha_i(\lambda),\qquad \forall i\in[n].
\]

It then recovers \(x_{ij}\) from \(y_{ij}\). For every good \(j\) with \(p_j(\lambda)>0\), the algorithm outputs each allocation fraction
\[
x_{ij}=
\begin{cases}
\dfrac{y_{ij}}{p_j(\lambda)}, & i\in T_j(F),\\
0, & i\notin T_j(F).
\end{cases}
\]

The constraint \(\sum_{i\in T_j(F)}y_{ij}=p_j(\lambda)\) implies
\[
\sum_i x_{ij}
=
\frac{1}{p_j(\lambda)}
\sum_{i\in T_j(F)}y_{ij}
=
1.
\]

\begin{remark}
When \(p_j(\lambda)=0\), the payment variables do not determine a unique allocation for good \(j\). In this case, any allocation supported on \(T_j(F)\) satisfies the equilibrium conditions, since the good generates zero payment.
\end{remark}

\paragraph{Completeness of the enumeration.}
Having established soundness—demonstrating that any feasible solution found by the algorithm yields a valid SPPE, we now prove completeness. Specifically, we show that our enumeration is exhaustive and captures all existing equilibria.

Suppose that there exists an SPPE \((\alpha,x)\). Define \(\lambda\) as in Eq.~\eqref{eq:lambda-definition}.

Since the range of every possible value of \(\lambda_j\) spans the entire line, \((\alpha,x)\) must lie in some cell \(F\) of the arrangement, and the algorithm eventually enumerates this cell. Within this cell, the sets \(E(F)\), \(L(F)\), and \(T_j(F)\) coincide with the unpaced buyers, the paced buyers, and the highest paced bidders in the equilibrium, respectively. For each good \(j\), choose a witness \(r_j\) whose paced bid realizes the second-price payment. If \(|T_j(F)|\geq 2\), such a witness may be chosen from \(T_j(F)\); if \(T_j(F)=\{w_j\}\), choose \(r_j\) to be a highest losing bidder. So \(r_j\in R_j(F)\), and our algorithm eventually enumerates the witness tuple \(r=(r_j)_{j\in[m]}\). Taking \( y_{ij}=x_{ij}p_j \) for the payments, it has to satisfy the cell, witness, and budget constraints in the corresponding linear system. As a result, every SPPE will be captured by at least one enumeration.

\paragraph{Time Complexity.}
We first establish a formal bound on the number of cells induced by our hyperplane arrangement, which is the key step in the running-time analysis.

\begin{lemma}
\label{lem:constant-goods-cells}
Given any instance of a second-price pacing game with a constant number \(c\) of goods, the number of cells $|\mathcal{F}|$ is polynomial in $n$. Moreover, all cells can be explicitly enumerated in polynomial time.
\end{lemma}

\begin{proof}
By construction, the hyperplane arrangement consists of two types. The first restricts individual coordinates:
\[
\lambda_j = v_{ij}, 
\quad 
\forall i\in [n],\, j\in [m] \text{ s.t. } v_{ij}>0,
\]
and the second restricts ratios between coordinates:
\[
\frac{\lambda_j}{\lambda_k} = \frac{v_{ij}}{v_{ik}}, 
\quad
\forall i\in [n], \forall j,k \in [m], j<k \text{ s.t. } v_{ij}, v_{ik}>0.
\]

There are at most \(nc\) hyperplanes of the first type and at most \(nc(c-1)/2\) hyperplanes of the second type. Therefore, the total number of hyperplanes \(H\) is bounded by:
\[
H
\leq
n\left(c+\frac{c(c-1)}{2}\right)
=
\frac{nc(c+1)}{2}.
\]

Because the number of goods \(c\) is constant, we have \(H = O(n)\).

By Theorem~\ref{thm:cut}, an arrangement of \(H\) affine hyperplanes in general position in \(\mathbb{R}^c\) induces $f_j$ faces of dimension $j$. Summing over all dimensions \(j=0,\dots,c\), the total number of faces is:
\[
\sum_{j=0}^{c} f_j
=
\sum_{j=0}^{c}
\left[
\binom{H}{c-j}
\sum_{\ell=0}^{j}
\binom{H-c+j}{\ell}
\right].
\]

Notably, while our specific arrangement may not be in general position, the general-position case strictly maximizes the number of induced faces. Because the sum is dominated by the highest-order terms, the total number of faces (which we collectively refer to as cells) simplifies asymptotically to:
\begin{equation}
\label{eq:cells-enumeration-complexity}
|\mathcal{F}| 
= O(H^c) 
= O(n^c).
\end{equation}

Since \(c\) is constant, the number of cells is strictly polynomial in \(n\). 

It remains to show that all these cells can be enumerated in polynomial time. We can achieve this by projecting the arrangement onto its 1D components. For each coordinate \(\lambda_j\), the input values \(\{v_{ij}\}_{i\in[n]}\) partition the positive real line into at most \(2n+1\) distinct regions (intervals and exact points). Similarly, for each of the \(c(c-1)/2\) pairs \(j<k\), the ratios \(\{v_{ij}/v_{ik}\}_{i\in[n]}\) partition the ratio line into at most \(2n+1\) regions. 

A combinatorial state is completely defined by choosing one such region for each of the \(c(c+1)/2\) projections. This yields at most \((2n+1)^{c(c+1)/2}\) possible states to enumerate, which is also polynomial in \(n\).
\end{proof}
\begin{remark}
While this naive combinatorial enumeration is exhaustive, it may over-generate mutually inconsistent states that do not define a non-empty geometric cell in \(\mathbb{R}^c\). However, this poses no computational difficulty. We simply construct the equilibrium linear program for every generated state. If a state corresponds to an empty cell, the resulting linear system will trivially be infeasible. Thus, the combined enumeration and verification process runs entirely in polynomial time.
\end{remark}

Thus, by Lemma~\ref{lem:constant-goods-cells}, the outer loop over cells has polynomially many iterations. Then for each fixed cell \(F\), the sets \(M_i(F)\), \(E(F)\), \(L(F)\), and \(T_j(F)\) can be constructed in polynomial time. The remaining nontrivial object needed by the algorithm is the collection of valid second-price witness sets \(R_j(F)\). These sets determine which buyers can serve as witnesses for the second-price payment of each good within the fixed cell \(F\). Since the main algorithm later enumerates witness tuples from \(\prod_{j\in[m]}R_j(F)\), we must ensure that these sets can also be computed efficiently. The following lemma establishes this point.

\begin{lemma}
\label{lem:second-price-polynomial}
Given any fixed cell \(F \in \mathcal{F}\), the valid second-price witness sets \(R_j(F)\) for all goods \(j\in [m]\) can be determined in polynomial time. 
\end{lemma}

\begin{proof}
If \(T_j(F)=\emptyset\), then no buyer attains the claimed highest bid level \(\lambda_j\), so the cell is infeasible. If \(|T_j(F)|\ge 2\), then the second price equals the highest paced bid \(\lambda_j\).
 In these cases, Algorithm~\ref{alg:second-price-witness} resolves trivially and terminates immediately.
It remains to consider the case where exactly one buyer attains the highest bid, meaning \(|T_j(F)|=1\). 

Suppose \(T_j(F)=\{w_j\}\). To find the valid second-price witnesses, the algorithm must test every other candidate \( r_j \in ([n] \setminus \{w_j\}) \cup \{0\}\). There are exactly \(n\) such candidates to enumerate. For each candidate \(r_j\), the algorithm constructs the following linear inequality system to enforce that \(r_j\) can validly hold the second-highest bid:
\[
\alpha_{r_j}(\lambda)v_{r_jj}
\geq
\alpha_i(\lambda)v_{ij},
\qquad
\forall i\in [n]\setminus\{w_j\}.
\]

This system contains at most \(n-1\) inequalities. As established previously, within the fixed cell \(F\), the multiplier expressions \(\alpha_i(\lambda)\) reduce to fixed linear functions of \(\lambda\). Thus, each inequality in the system is strictly linear. The feasibility of this system, combined with the geometric boundaries of cell \(F\), can be verified in polynomial time using standard linear programming.

Consequently, Algorithm~\ref{alg:second-price-witness} runs in polynomial time for any single good. Because our main procedure, Algorithm~\ref{alg:constant-goods}, invokes this subroutine exactly \(c\) times per cell (once for each good), the total time required to determine the complete witness sets \(R_j(F)\) for all \(j \in [m]\) within a fixed cell remains strictly polynomial in \(n\).
\end{proof}

By Lemma~\ref{lem:second-price-polynomial}, the second-price witness sets \(R_j(F)\) can also be computed in polynomial time.

For this cell, the algorithm enumerates witness tuples
\[
r=(r_j)_{j\in[m]}\in \prod_{j\in[m]}R_j(F).
\]

Since \(|R_j(F)|\le n\) and \(m=c\), there are at most
\begin{equation}
\label{eq:sp-witness-enumeration-complexity}
\prod_{j\in[m]}|R_j(F)|= O(n^c)
\end{equation}
tuples.

For each pair \((F,r)\), the algorithm solves one linear feasibility problem with \(O(n)\) variables and \(O(n)\) constraints, since the number of goods \(c\) is constant. By the linear-programming algorithm of \citet{CLS:19}, an LP with \(N\) variables can be solved, to relative accuracy \(\delta\), in expected time \(O^*\!\left(N^\omega \log\frac{N}{\delta}\right)\), 
where \(\omega\approx 2.37\) is the matrix multiplication exponent. Therefore, in our setting each LP can be solved in expected time

\begin{equation}
\label{eq:LP-complexity}
O^*\!\left(n^\omega \log\frac{n}{\delta}\right). 
\end{equation}

Combining Eqs.~\eqref{eq:cells-enumeration-complexity}, \eqref{eq:sp-witness-enumeration-complexity} and \eqref{eq:LP-complexity}, the total expected running time is
\[
O(n^c)\cdot O(n^c)\cdot 
O^*\!\left(n^\omega \log\frac{n}{\delta}\right)
=
O^*\!\left(n^{2c+\omega}\log\frac{n}{\delta}\right).
\]

Since \(c\) is constant, this is polynomial in \(n\), and hence Algorithm~\ref{alg:constant-goods} runs in polynomial time.
\end{proof}

\section{Extension to a Nonconstant Number of Goods}

We now show that the constant-goods result can be extended to instances with a nonconstant number of goods, as long as the number of distinct good types is constant.

\begin{definition}[Good types]
Two goods \(j\) and \(j'\) are said to have the same {\em type} if they have the same valuation profile across all buyers, that is,
\[
v_{ij}=v_{ij'}, \qquad \forall i\in[n].
\]

This relation partitions the set of goods into equivalence classes, which we call {\em good types}.
\end{definition}

\begin{corollary}
\label{cor:constant-good-types}
Given any instance of a second-price pacing game in which both the number of goods and the number of buyers are nonconstant, if the number of good types is constant, an SPPE can be computed in polynomial time.
\end{corollary}

\begin{proof}
Let \(\mathcal T\) be the set of good types, and let
\[
|\mathcal T|=c,
\]
where \(c\) is a fixed constant, by definition. For each type \(\tau\in\mathcal T\), let \(S_\tau\) denote the set of goods of type \(\tau\). By definition, for any \(j,j'\in S_\tau\),
\[
v_{ij}=v_{ij'}, \qquad \forall i\in[n].
\]

We now construct an aggregated instance with one aggregate good for each type \(\tau\in\mathcal T\). For every buyer \(i\), we define its valuation for the aggregate good \(\tau\) by
\[
\widehat v_{i\tau}
:=
\sum_{j\in S_\tau} v_{ij}.
\]

Since all goods in \(S_\tau\) have the same 
\[
\widehat v_{i\tau}
=
|S_\tau|\,v_{ij},
\qquad \text{for any } j\in S_\tau.
\]

Now the aggregated instance has exactly \(c\) goods. Therefore, by Theorem~\ref{thm:main}, we can compute an SPPE of the aggregated instance in polynomial time. Let this equilibrium be \( (\widehat\alpha,\widehat x)\). We now expand this SPPE back to original. For every buyer \(i\), keep the same pacing multiplier:
\[
\alpha_i=\widehat\alpha_i.
\]

For every original good \(j\in S_\tau\), we let
\[
x_{ij}=\widehat x_{i\tau},
\qquad \forall i\in[n].
\]

Each original good is fully allocated, because
\[
\sum_i x_{ij}
=
\sum_i \widehat x_{i\tau}
=
1.
\]

It remains to verify that the recovered pair \((\alpha,x)\) is an SPPE of the original instance.

Fix a type \(\tau\in\mathcal T\) and a good \(j\in S_\tau\). For every buyer \(i\), we have
\[
\widehat\alpha_i\widehat v_{i\tau}
=
\widehat\alpha_i |S_\tau| v_{ij}
=
|S_\tau| \alpha_i v_{ij}.
\]

The ordering of paced bids on the aggregate good \(\tau\) is exactly the same as the ordering of paced bids on every original good \(j\in S_\tau\). In particular, the set of highest paced bidders and the set of second-price witnesses are preserved.

Let \(\widehat p_\tau\) be the true payment of the aggregate good \(\tau\), and let \(p_j\) be the true payment of any original good \(j\in S_\tau\). Since all paced bids in the aggregate instance are multiplied by \(|S_\tau|\), we have
\[
\widehat p_\tau
=
|S_\tau|p_j.
\]

Now consider the total payment of buyer \(i\) from goods of type \(\tau\) in the original instance. By the construction of \(x\), we have
\[
\sum_{j\in S_\tau} x_{ij}p_j
=
\sum_{j\in S_\tau} \widehat x_{i\tau}p_j
=
\widehat x_{i\tau}|S_\tau|p_j
=
\widehat x_{i\tau}\widehat p_\tau.
\]

This is exactly buyer \(i\)'s payment for the aggregate good \(\tau\) in the aggregated instance. So every buyer has the same total payment in the original instance as in the aggregated instance.

As a result, budget feasibility is preserved. Moreover, if a buyer is paced in the aggregated equilibrium, then it exhausts its budget there, and also exhausts its budget after expansion. If a buyer does not exhaust its budget, its pacing multiplier is \(1\) in the aggregated equilibrium, and the same multiplier is used in the original instance. Thus the ``no unnecessary pacing'' condition is also preserved.

Finally, allocation in the aggregated instance is supported only on highest paced bidders. The highest paced bidder sets are preserved for every original good, and the recovered allocation also assigns goods only to highest paced bidders.

As proof above, \((\alpha,x)\) is an SPPE of the original
instance. This equilibrium can be computed in polynomial time.
\end{proof}

\section{Conclusion}
In this paper, we presented a polynomial-time algorithm for computing an exact second-price pacing equilibrium in markets with a constant number of goods, successfully resolving an open question raised by \citet{CL:25} and \citet{yan:26}. Furthermore, our \(\lambda\)-space reformulation offers an analytical tool for mapping pacing multipliers to bid levels, effectively bypassing traditional dimensionality bottlenecks. 

Given the known PPAD-hardness of computing exact equilibria in general, unrestricted markets, several critical questions remain. Most notably, while \citet{CL:25} established a \(1/3\)-inapproximability bound, does there exist a polynomial-time algorithm for computing a non-trivial constant-factor approximate SPPE (e.g., a \(0.99\)-approximation) in general markets? More ambitiously, what is the tight bound?

\section{Acknowledgments}
This work was supported by the National Natural Science Foundation of China (Grants 62472029 and 62572476) and the Key Laboratory of Interdisciplinary Research of Computation and Economics (Shanghai University of Finance and Economics), Ministry of Education.
\bibliography{aaai2026}
\clearpage

\makeatletter
\@ifundefined{isChecklistMainFile}{\newif\ifreproStandalone\reproStandalonefalse}{}
\makeatother

\end{document}